\documentclass[a4paper,10pt]{article}
\usepackage{graphicx}
\usepackage{a4wide}
\usepackage{hyperref}

\title{Learning Exactly Linearizable Deep Dynamics Models}
\author{Ryuta Moriyasu\thanks{Green Fuel Research-Domain, Toyota Central R\&D Labs., Inc., 41-1 Yokomichi, Nagakute, Aichi 480-1192, JAPAN} \and Masayuki Kusunoki\thanks{Engineering Dept., Engine Division, Toyota Industries Corporation, 3 Hama-cho, Hekinan-Shi, Aichi 447-8507, JAPAN} \and Kenji Kashima\thanks{Graduate School of Informatics, Kyoto University, Yoshida-honmachi, Sakyo-ku, Kyoto 606-8501, JAPAN}{ }\thanks{Corresponding author (e-mail: kk@i.kyoto-u.ac.jp, tel: (+81) 75-753-5512)}}
\date{}

\usepackage{natbib}
\bibpunct{(}{)}{;}{a}{,}{,}
\usepackage{times}
\usepackage{amsmath,amssymb, mathrsfs}
\usepackage{amsthm}
\usepackage{bm}
\usepackage{bbm}
\usepackage{tikz}
\usepackage{comment}
\usepackage{mathtools}
\usepackage{slashbox}
\usetikzlibrary{calc}
\usetikzlibrary{positioning}
\usetikzlibrary{arrows.meta}

\newtheorem{dfn}{Definition}

\newtheorem{thm}[dfn]{Theorem}

\newcommand{\del}{\partial}
\newcommand{\diff}[2]{\frac{\mathrm{d} #1}{\mathrm{d} #2}}
\newcommand{\pdiff}[2]{\frac{\del #1}{\del #2}}

\newcommand{\PAR}[1]{{\left( #1 \right)}}

\newcommand{\R}{\mathbb{R}}
\newcommand{\T}{\mathrm{T}} 

\newcommand{\argmin}{\mathop{\mathrm{arg~min}}\limits}
\newcommand{\minimize}{\mathop{\mathrm{minimize}}\limits}

\newcommand{\mat}[1]{\begin{bmatrix}#1\end{bmatrix}}

\newenvironment{keywords}
{\bgroup\leftskip 20pt\rightskip 20pt \small\noindent{\bfseries
Keywords:} \ignorespaces}%
{\par\egroup\vskip 0.25ex}

\makeatletter
\newcommand{\figcaption}[1]{\def\@captype{figure}\caption{#1}}
\newcommand{\tblcaption}[1]{\def\@captype{table}\caption{#1}}
\makeatother

\begin{document}

	\maketitle

	\begin{abstract}%
	Research on control using models based on machine-learning methods has now shifted to the practical engineering stage. Achieving high performance and theoretically guaranteeing the safety of the system is critical for such applications. In this paper, we propose a learning method for \emph{exactly linearizable} dynamical models that can easily apply various control theories to ensure stability, reliability, etc., and to provide a high degree of freedom of expression. As an example, we present a design that combines simple linear control and control barrier functions. The proposed model is employed for the real-time control of an automotive engine, and the results demonstrate good predictive performance and stable control under constraints.
	
	\end{abstract}
	
	
	\section{Introduction}\label{section: introduction}
	
	In recent years, there has been a growing interest in using machine learning (particularly deep learning) for modeling dynamical systems \citep{kocijian01, Hedjar2013, Lenz2015, Moriyasu2019}.
	Unlike the traditional modeling approach, machine learning can be used to create highly accurate models more easily without the need for detailed domain knowledge. However, control design using such models can become challenging because of the strong nonlinearity of the machine-learning models.
	Model predictive control (MPC) is a typical control method that can handle nonlinearity, wherein the model is used to optimize the predicted future behavior of the system; however, the optimal control problem solved each time is often nonconvex \citep{awrynczuk2008, nghiem01, Gros2019} when the model is highly nonlinear, thereby making it difficult to ensure the uniqueness and optimality of the numerically obtained solution and the continuity of the control law \citep{Moriyasu2022}.
	Various nonlinear control theories besides MPC cannot be easily applied to general machine learning models because the model structures to which the theories can be applied are often limited to specific systems, such as input-affine systems.
	
	Machine learning models with special structures for controlling design methods have been proposed to overcome the abovementioned problem. For example, a dynamical model structure called the input convex recurrent neural network (ICRNN) \citep{Chen2019a} was proposed to guarantee the convexity of the optimal control problem in economic MPC, i.e., MPC that directly aims to minimize or maximize the outputs or states. This is an extension of the idea of an input convex neural network (ICNN) \citep{Amos2017}, which is a deep neural network that can guarantee the convexity of the input--output relationships, to dynamical systems. Further, a structured Hammerstein--Wiener (S--HW) model \citep{Moriyasu2022} can ensure the convexity of the usual MPC problems for reference tracking or regulation for which convexity cannot be guaranteed with ICRNN. This is an application of the ICNN and bijective neural network (BNN) \citep{Baird2005} to the Hammerstein--Wiener (HW) model \citep{hammerstein, wiener}, which has long been known in the field of system identification. The HW model has a linear dynamical system sandwiched between static nonlinear bijective mappings. 
	Static nonlinearity can be canceled out by utilizing bijectivity, and the control problems can be reduced to linear control problems \citep{Fruzzetti1997, CERVANTES2003}.
	In the S--HW model, BNNs are used to express bijective mappings and an ICNN is attached to learn the additional outputs to be constrained. Therefore, the constrained nonlinear tracking control problems can be reduced to convex problems. 
	
	However, such a specific structured model can lose its expressive ability when the above properties are ensured by limiting the range of weights, choice of activation function, and connecting topology of the network. 
	Therefore, it is necessary to propose a model structure that can learn a wider range of objects with high accuracy and is easy to use in control design.
	
	Given this context, we focused on a class of nonlinear systems called \emph{exactly linearizable} (EL) systems and proposed a learning model that is guaranteed to have this property. Although the HW model represents the input/output of a nonlinear system by transforming the input/output of a linear dynamical system using static nonlinear functions, EL systems can be attributed to linear dynamics using nonlinear dynamic feedback \citep{Khalil2002}.
	This feedback is an extension of static mapping, and therefore, an EL system can be considered an extension of the HW system. 
	The EL models lose the ability to make the constrained MPC convex because of the nonlinear feedback structure. However, because they are linearizable, linear control methods, such as linear quadratic regulators (LQR), and various nonlinear control methods can be easily applied with control barrier functions (CBF), which enables the constraint-aware control design.
	
	The contribution of this study is that it proposes an approach to learn EL models, which is a class of systems that have a higher degree of freedom of expression than conventional methods, while facilitating control design. 
	Verifying that a dynamical system is EL and constructing a dynamic feedback law that achieves linearization is usually based on heuristics and is therefore impractical. Our idea of introducing output feedback in a static nonlinear mapping of the S--HW model structure can guarantee that the learning model has EL properties and can explicitly obtain a dynamic feedback law.
	
	The remainder of this paper is organized as follows: Section 2 shows the structure of the proposed model in comparison with the S--HW model.
	Section 3 presents a control design method using the proposed model. Further, we present the simplest method based on LQR. 
	Further, this section introduces integral CBFs \citep{Ames2021} which enable the model to be more flexible in terms of its degrees of freedom.
	In Section 4, the model is applied to the air path and combustion process of an engine, which is a complex nonlinear system, to verify its advantages in terms of accuracy and show that control objectives can be successfully achieved via a control simulation.
	
	\section{Exactly Linearizable Model}\label{section: modeling}
	
	\subsection{Problem Statement}\label{problem}
	
	This study addresses the modeling of nonlinear MIMO systems and the design of control based on the obtained models considering the upper and lower input constraints and the upper output constraint. Let $u\in\R^{m}$, $d\in\R^{l}$, $x\in\R^{n}$, and $y\in\R^{p}$ represent the control input, exogenous input (which can include measurable disturbance), state, and output, respectively.
	We assume that the state $x$, exogenous input $d$, and their time-derivatives $\dot{x}, \dot{d}$ (numerical differentiation is acceptable) can be observed online and $d$ cannot be adjusted.
	Observation noise is not considered explicitly in this paper.

	We consider building a model that can predict the response $x,y$ to $u,d$ by deep learning. 
	In control design, the controller is designed to maintain $x$ close to the target $x_\mathrm{d}\in\R^{n}$ considering the upper and lower constraints for each element of the control input $u=[u_1,\ldots,u_{m}]^\T$
	\begin{align}
		\underline{u}_i  \leq u_i \leq \overline{u}_i \ (i=1,\ldots,m)\label{cons_v}
	\end{align}
	and the upper constraint for each element on the output $y=[y_1,\ldots,y_{p}]^\T$
	\begin{align}
		y_j \leq \overline{y}_j \ (j=1,\ldots,p)\label{cons_z}.
	\end{align}
	
	\subsection{Exact Linearization}\label{subsec:exact_linearization}
	
	Exact linearization is entirely different from standard linearization based on the Taylor expansion at a specific operating point. 
	An input affine nonlinear system is described by 
	\begin{align}
		\dot{x} = \overline{f}(x) + \overline{g}(x)u \label{eq: exact linearizable equation}
	\end{align}
	with $x \in \R^{n}, u \in \R^{m}$.
	In the exact linearization method, system \eqref{eq: exact linearizable equation} is transformed by using coordinate transformation and nonlinear feedback, which are denoted by
	\begin{align}
		\xi &= \Phi(x) \label{eq: coordinate change} \\
		u &= \Psi_1(x) + \Psi_2 (x)v \label{eq: exact linearization feedback}
	\end{align}
	with a bijective map $\Phi$, nonlinear functions $\Psi_1,\Psi_2$, and newly introduced input $v\in \R^m$. This yields
	\begin{align}
		\dot{\xi} &=
		\frac{\partial \Phi}{\partial x}\{\overline{f}(\Phi^{-1}(\xi)) + \overline{g}(\Phi^{-1}(\xi)) \Psi_1(\Phi^{-1}(\xi))\}+ \frac{\partial \Phi}{\partial x}\overline{g}(\Phi^{-1}(\xi))\Psi_2(\Phi^{-1}(\xi))v \nonumber \\
		&=: f(\xi) + g(\xi)v.
	\end{align}
	Then, if we can find the coordinate transformation \eqref{eq: coordinate change} and feedback \eqref{eq: exact linearization feedback} such that
	\begin{align}
		f(\xi) = A\xi, \ g(\xi) = B
	\end{align}
	and $(A, B)$ are controllable, the system is considered \emph{exactly linearizable}. 
	
	Assuming that the plant and feedback are input-affine, as in \eqref{eq: exact linearizable equation} and \eqref{eq: exact linearization feedback}, such a linearizing transformation and feedback exist under some theoretically sufficient conditions; see Appendix~1. 
	However, this input-affine structure is sometimes restrictive for the dynamics of complex real systems, and therefore, we assumed the following target system, coordinate transformation, and feedback:
	\begin{align}
		\dot{x} = \overline{f}(x,u),\ \xi = \Phi(x), \ u = \Psi^{-1}(v, x)
	\end{align}
	for which the dynamics of $\xi$ are affine such that
	\begin{align}\label{eq:x_affine}
		\dot{\xi} =& A\xi + Bv + c. 
	\end{align}
	This motivated us to consider a continuous-time system represented by the following dynamics.
	
	\begin{dfn}\label{dfn:exatly_linearizable}
		Suppose $\Psi: \R^{m}\times \R^{n}\times \R^{l}\rightarrow \R^{m}$, $\Phi: \R^{n}\times \R^{l}\rightarrow \R^{n}$ and $\Xi : \R^{n}\times \R^{m}\times \R^{l}\rightarrow \R^{p}$ satisfy the following conditions:
		$\Psi( \cdot, x, d ),\ \Phi( \cdot, d)$  are bijective.
		$\Xi( \cdot, \cdot, d) $ is convex for any $x, d$,
		then, the system
		\begin{align}
			v(t) &= \Psi(u(t),x(t), d(t)),  \label{eq: v to u}\\
			\dot{\xi}(t) &= A(d(t)) \xi(t) + B(d(t)) v(t) + c(d(t)), \label{eq: dotx on d}\\
			x(t) &= \Phi^{-1}(\xi(t),d(t)), \label{eq: x to y}\\
			y(t) &= \Xi (\xi(t),v(t),d(t)) \label{eq: x,u to z}
		\end{align}
		is denoted as an exactly linearizable model, where $t\in\R$ and $v\in\R^{m}$ and $\xi\in\R^{n}$ represent the continuous time and the internal input and state of the affine dynamics, respectively. 
	\end{dfn}
	Note that the drift term $c$ in (\ref{eq:x_affine}) and (\ref{eq: dotx on d}) is introduced to extend the degree of freedom of the model and is removable if it is unnecessary.
	
	\subsection{Exactly Linearizable Model}
	
	In this section, we propose a neural network based model structure that corresponds to the Definition~\ref{dfn:exatly_linearizable}, which we call the EL model.
	A block diagram of the EL model is shown in Fig.~\ref{fig: Exactly linearizable model}. 
	The red line in Fig.~\ref{fig: Exactly linearizable model} shows its difference from the S--HW model \citep{Moriyasu2022}.
	The S--HW model refers to the model in which \eqref{eq: v to u} is restricted to the output-independent $v(t) = \Psi(u(t), d(t))$. 
	This simple extension in the EL model allows the model to learn a wider range of dynamical systems than the S--HW model. 
	For example, the S--HW model always has a unique equilibrium corresponding to each input, whereas the EL model can represent systems with multiple equilibria. This difference greatly improves the accuracy of the model but requires more careful treatment during control design. 
	\begin{figure}[t]
		\centering
		\includegraphics[width=.6\linewidth]{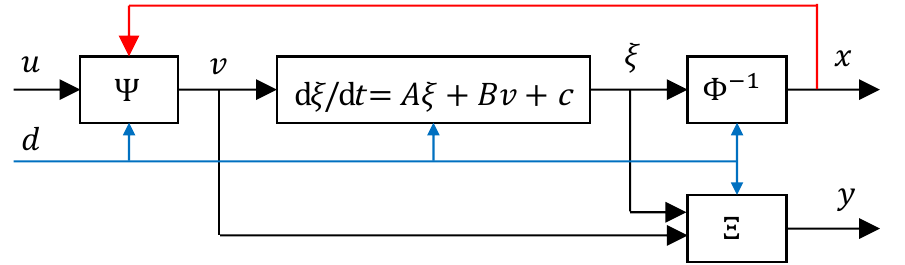}
		\caption{Exactly linearizable model}
		\label{fig: Exactly linearizable model}
	\end{figure}	
	\begin{figure}[t]
		\centering
		\includegraphics[width=1.0\linewidth]{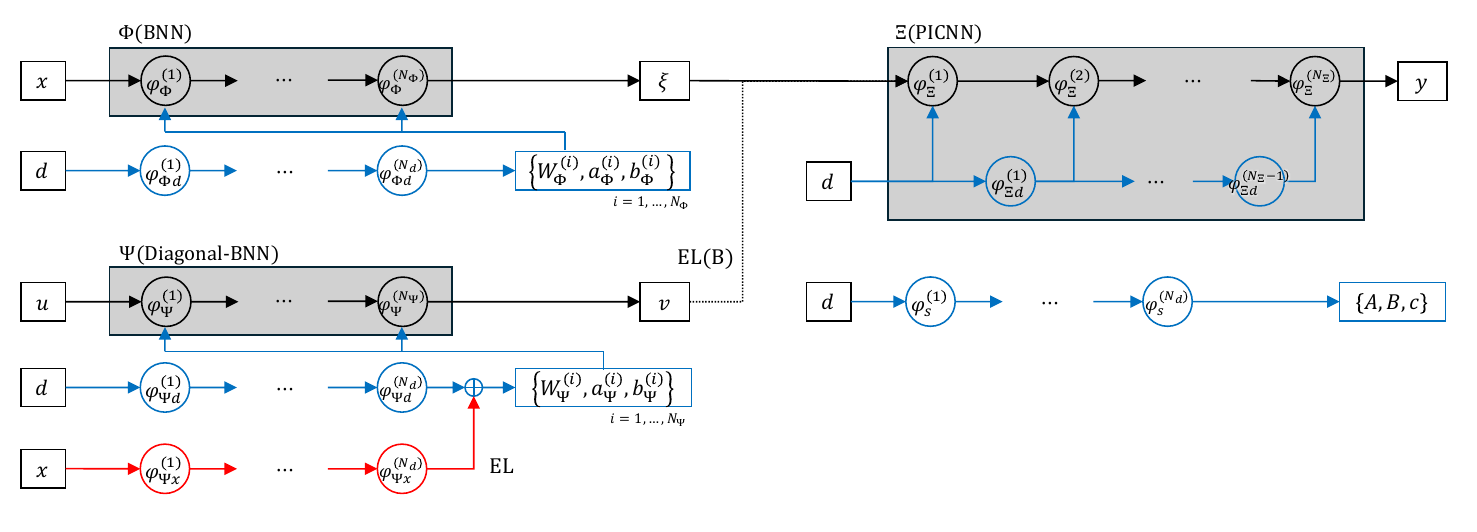}
		\caption{Network structure of exactly linearizable model}
		\label{fig: Network}
	\end{figure}

	The whole network structure of EL models is illustrated in Fig.~\ref{fig: Network}.
	Recall that functions $\Phi(\cdot,d)$ and $\Psi(\cdot,x,d)$ must be bijective mappings. 
	We parameterize each mapping using a bijective neural network (BNN) \citep{Baird2005} and a diagonal-BNN, respectively.
	$\Phi$ is represented as the $N_\Phi$-layered BNN
	\begin{align}
		\Phi(x,d) = \varphi^{(N_\Phi)}_\Phi \circ \varphi^{(N_\Phi-1)}_\Phi \circ \cdots \circ \varphi^{(1)}_\Phi(x,d),
	\end{align}
	where $\varphi^{(i)}_\Phi(\cdot,d) (i=1,\ldots,N_\Phi)$ are parametrized bijective mappings. 
	In this study, we employ 
	\begin{align}
		\varphi^{(i)}_\Phi(x,d) = \sinh^{-1}(a^{(i)}_\Phi(d)+\sinh(W^{(i)}_\Phi(d) x+b^{(i)}_\Phi(d)))\ (i=1,\ldots,N_\Phi),
	\end{align}
	where $W^{(i)}_\Phi(d)$ and $a^{(i)}_\Phi(d),b^{(i)}_\Phi(d)$ represents matrix-valued and vector-valued functions to be learned, described later.
	Similarly, $\Psi$ is represented as the $N_\Psi$-layered diagonal BNN 
	\begin{align}
		\Psi(u,x,d) &= \varphi^{(N_\Psi)}_\Psi \circ \varphi^{(N_\Psi-1)}_\Psi \circ \cdots \circ \varphi^{(1)}_\Psi(u,x,d),\\
		\varphi^{(i)}_\Psi(u,x,d) &= \sinh^{-1}(a^{(i)}_\Psi(x,d)+\sinh(W^{(i)}_\Psi(x,d) u + b^{(i)}_\Psi(x,d))) \ (i=1,\ldots,N_\Psi).
	\end{align}
	Unlike $\Phi$, functions $W^{(i)}_\Psi \ (n=1,\cdots,N_\Psi)$ have to be diagonal so that the input--output relationship of the BNN becomes element-wise.
	Moreover, the function $\Xi(\cdot,\cdot,d)$ must be convex, and we adopt a partially input convex neural network (PICNN) \citep{Amos2017} for parametrization.
	PICNN is represented as the $N_\Xi$-layered neural network:
	\begin{align}
		\Xi(\xi,v,d)  &= \zeta^{(N_\Xi)}, \zeta^{(0)} = [\xi^\T, v^\T]^\T, \eta^{(0)} = d, \\
		\zeta^{(i)} &= \varphi_{\Xi}^{(i)} \left(
			   W_\zeta^{(i)} \PAR{ \zeta^{(i-1)} \odot {\rm softplus}(W_{\zeta\eta}^{(i)}\eta^{(i-1)}+b_{\zeta\eta}^{(i)}) } \right.\nonumber\\
			 		& \ \ \ \ \ \ \ \ \  \left. + \ W_0^{(i)} \PAR{\zeta^{(0)} \odot (W_{0\eta}^{(i)} \eta^{(i-1)} +b_{0\eta}^{(i)}) } 
			 		  + W_{b\eta}^{(i)} \eta^{(i-1)} + b_{b\eta}^{(i)} \right) ,\\
		\eta^{(i)} &= \varphi_{\Xi d}^{(i)} \PAR{ W_\eta^{(i)} \eta^{(i-1)} + b_\eta^{(i)} }, (i=1,\ldots,N_\Xi),
	\end{align}
	where $W_{\zeta}^{(i)},W_0^{(i)},W_\eta^{(i)},W_{\zeta\eta}^{(i)},W_{0\eta}^{(i)},W_{b\eta}^{(i)}\ (i=1,\ldots,N_\Xi)$ and $b_\eta^{(i)},b_{\zeta\eta}^{(i)},b_{0\eta}^{(i)},b_{b\eta}^{(i)}\ (i=1,\ldots,N_\Xi)$ are weights and biases to be learned, $\varphi_\Xi^{(i)},\varphi_{\Xi d}^{(i)}\ (i=1,\ldots,N_\Xi)$ are activation functions, and ${\rm softplus}(x):= \log(1+e^x)$. 
	$\varphi_\Xi^{(i)}, W_\zeta^{(i)} \ (i=1,\ldots,N_\Xi)$ have to be monotonically non-decreasing convex functions and nonnegative matrices, respectively, to ensure convexity of $\Xi$ with respect to $\xi, v$.

	In this model, $d$-dependent functions $A, B, c, W^{(i)}_\Phi, a^{(i)}_\Phi, b^{(i)}_\Phi, W^{(j)}_\Psi, a^{(j)}_\Psi, b^{(j)}_\Psi \ (i=1,\ldots,N_\Phi, \ j=1,\ldots,N_\Psi)$ can be arbitrary functions; however, in this study, a standard fully connected neural network (FNN) is used to represent each element.
	Precisely, $W^{(j)}_\Psi, a^{(j)}_\Psi, b^{(j)}_\Psi \ (j=1,\ldots,N_\Psi)$ also have $x$-dependency in EL model.
	We deal with this $(x,d)$-dependency by employing $x$-dependent FNN and add its output with that of $d$-dependent FNN to obtain the above functions (shown as red in Fig.~\ref{fig: Network}).
	The example of detailed settings is show in Section~4.
	
	A summary of the superior properties of the EL model is presented below. First, the dynamical system can be considered linear because of the dynamic feedback. Second, the reference for $x$ can be replaced with that for $\xi$ because of the bijectivity of $\Phi$. 
	Third, the admissible set of $v$ corresponding to the input constraint \eqref{cons_v} is convex for any $\xi$ owing to the element-wise bijectivity of $\Psi$.
	Finally, the admissible set of $\xi,v$ corresponding to the output constraint \eqref{cons_z} is convex because of the convexity of $\Xi$.
	
	The learning process of the EL model is as follows: 
	The EL model structure is similar to that of the S--HW model \citep{Moriyasu2022}, except that output $x$ is fed back into \eqref{eq: v to u}, to use the same procedure for learning the model. The outputs predicted by the model are expressed with a hat to distinguish it from real data. From \eqref{eq: v to u}--\eqref{eq: x,u to z}, we obtain
	\begin{align}
		\dot{\hat{x}} &= \PAR{\pdiff{\Phi}{x}}^{-1} \PAR{A(d) \Phi(x,d) + B(d) \Psi(u,x,d) + c(d) - \pdiff{\Phi}{d} \dot{d}},\\  \label{eq: y hat dot}
		\hat{y} &= \Xi (\Phi(x,d), \Psi(u,x,d), d).
	\end{align}
	A key point is that the bijectivity of the BNN enables these evaluations without using internal signals $\xi, v$. 
	The data required to evaluate these are those of $u,x,d,\dot{d}$.
	The prediction error can be evaluated as $e:=[(\dot{\hat{x}}-\dot{x})^\T, (\hat{y}-y)^\T]^\T$ using the answer data of $\dot{x},y$.
	Thus, model learning is reduced to an error minimization problem 
	\begin{align}
		\minimize_\theta \frac{1}{N_s} \sum_{i=1}^{N_s} e_i^\T Q_e e_i,
	\end{align}
	where $N_s$, $e_i$, $Q_e$, and $\theta$ represent the number of sampled data points, prediction error for the $i$~th data point, positive definite matrix, and model parameter vector, respectively.
	$\theta$ includes all free parameters, i.e., weight matrices ans bias vectors, of the network (Fig.~\ref{fig: Network}) and is not listed here due to the sheer number of elements.
	The above problem can be solved typically by stochastic gradient descent (SGD) with a machine learning framework such as Tensorflow.
	Using such tools, Jacobians $\del \Phi/\del x, \del \Phi/\del d$ and also their derivatives, which are required in gradient based learning algorithms, can be caluculated automatically.
	
	Note that, (\ref{eq: y hat dot}) is used only in learning phase. 
	In control computation, we can use (\ref{eq: dotx on d}) and (\ref{eq: x to y}) to evaluate output $y$.
	That allows us to avoid heavy computation for evaluating Jacobian and achieve fast computation of model-based controller.
	In addition, data, especially of time derivatives ($\dot{x},\dot{d}$), are often noisy in practice, but in deep learning, a certain amount of noise may be rather beneficial to improve generalization performance \citep{liu2020does}. 
	If simulators or very accurate experimental instruments are used to generate data, we recommend adding some noise to avoid overfitting.

	\section{Control Design}\label{section: control design}
	
	In this section, we describe a control method that satisfies constraints \eqref{cons_v} and \eqref{cons_z} and regulates the output $x$ to the reference value $x_\mathrm{d}$.
	This constrained regulation control can be designed as a convex MPC using the S--HW model as shown by \cite{Moriyasu2022}; however, the introduction of the output feedback in \eqref{eq: v to u} does not ensure the convexity of the input admissible set in the finite-horizon optimal control problem.
	In this paper, we propose an alternative design method for an EL model that avoids this problem and attributes the control law to convex optimization instead.
	
	\subsection{Regulation Control}
	\label{subsec:LQR}
	
	From the bijectivity of $\Phi$ and $\Psi$ in Definition \ref{dfn:exatly_linearizable}, the regulation control problem for $x$ (using input $u$) can be reduced to a linear control problem for $\xi$ (using input $v$) for which the optimal design methods are well developed.
	For the simplicity of discussion, exogenous input $d$ is assumed to be constant in the transformed linear dynamics in \eqref{eq: v to u}–\eqref{eq: x to y}, such that
	\begin{align*}
		\dot{\xi} &= A (\bar{d})\xi + B(\bar{d})v + c(\bar{d}) ,\\
		\xi &= \Phi(x, \bar{d}), \ 
		v = \Psi(u, x, \bar{d}).
	\end{align*}
	Suppose that $x_\mathrm{d}$ is realizable as a steady state with $d = \bar{d}$ and there exists $v_\mathrm{d}$ such that 
	\begin{align}\label{eq: target value on LQR}
		A(\bar{d}) \xi_\mathrm{d} + B(\bar{d}) v_\mathrm{d} + c(\bar{d}) =0,\ 
		\xi_\mathrm{d} := \Phi( x_\mathrm{d}, \bar{d}).
	\end{align}
	We consider the control input $v(t)$ for the linear part that minimizes the cost function 
	\begin{align}
		J := \int_{0}^{\infty} \tilde{\xi}^\T Q\tilde{\xi} + \tilde{v}^\T R \tilde{v} \ \mathrm{d}t 
	\end{align}
	where $\tilde{\xi} := \xi - \xi_\mathrm{d},\ \tilde{v} := v - v_\mathrm{d}$.
	The optimal $v$ is given by a linear quadratic regulator
	\begin{align}\label{LQR}
		v = v_\mathrm{d} + K\tilde{\xi}, \ K \coloneqq  -R^{-1}B^\T P
	\end{align}
	where $P$ represents a unique positive definite solution for the matrix Riccati equation
	\begin{align}
		PA + A^\T P - PBR^{-1}B^\T P + Q = 0. \label{eq: riccati equation}
	\end{align}
	The control input $ u (t) $ is obtained by performing an inverse transformation of \eqref{eq: v to u} for the control input $v(t) $ in \eqref{LQR}.

	Since $\xi$ and $v$ have no physical meaning, determining $Q$ and $R$ to achieve desirable performance is generally difficult.
	Therefore, it is desirable to be able to consider another objective function $J_{\rm r} = \int_0^\infty \tilde{x}^\T Q_{\rm r} \tilde{x} + \tilde{u}^\T R_{\rm r} \tilde{u} \ {\rm d}t$, where $\tilde{x} := x - x_{\rm d}, \tilde{u} := u - u_{\rm d}, u_{\rm d} := \Psi^{-1}(v_{\rm d}, x_{\rm d},\bar{d})$.
	This function is difficult to minimize directly, but for example, linearization $\tilde{x} \approx (\del\Phi/\del x)|_{(x_{\rm d},d)}^{-1} \tilde{\xi}, \tilde{u} \approx (\del\Psi/\del u)|_{(u_{\rm d},x_{\rm d},d)}^{-1} \tilde{v}$ yields $J_{\rm r} \approx J$ around the target point $x_{\rm d},u_{\rm d}$ with $Q = (\del\Phi/\del x)^\T Q_{\rm r} (\del\Phi/\del x), R = (\del\Psi/\del u)^{\T} R_{\rm r} (\del\Psi/\del u)$.
	This can assist designing parameters $Q,R$. 

	The overall flow of the calculation of the control input $u(t)$ is shown in Fig.~\ref{fig: Controller implementation}. 
	The controller calculates the control input $u(t)$ from the output signal $x(t)$ obtained from the system to be controlled and the desired target value $x_\mathrm{d},u_\mathrm{d}$ satisfying (\ref{eq: target value on LQR}) and $u_{\rm d} = \Psi^{-1}(v_{\rm d}, x_{\rm d},\bar{d})$.
	In the controller, the output signal $x(t)$ is converted into the state variable $\xi(t)$ through the bijective map $\Phi$, and the internal input signal $v(t)$, which is linearly related to $\xi(t)$, is obtained using the linear control method \eqref{LQR}. 
	In addition, the control input $u(t)$ is obtained by mapping the internal input signal $v(t)$ using the bijective map $\Psi^{-1}$.	
	Fig.~\ref{fig: Controller design} shows the resulting closed-loop system comprising Figs.~\ref{fig: Exactly linearizable model} and \ref{fig: Controller implementation}, which perform feedback control on the linear dynamics $\dot{\xi}=A\xi+Bv+c$.
	The structure is only a linear system in which the bijective maps $\Psi^{-1}$ and $\Phi$ on both sides disappear. 

	The above design is based on unconstrained linear-quadratic regulator (LQR) and is one of the simplest methods that takes advantage of the internal linearity of the EL model. 
	For a more advanced control, nonlinear optimal control theory can also be employed (see Appendix~2).

	\begin{figure}[t]
		\begin{minipage}[b]{.55\textwidth}
			\centering
			\includegraphics[width=1\linewidth]{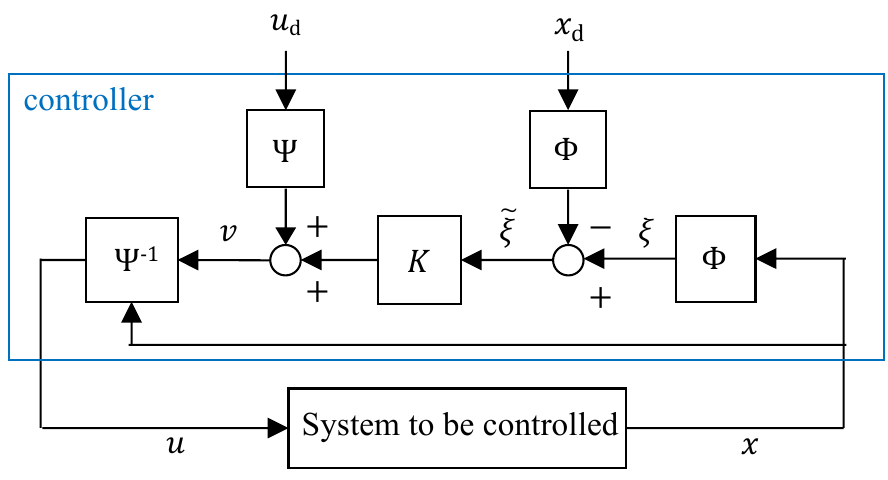}
			\caption{Controller implementation}
			\label{fig: Controller implementation}		
		\end{minipage}
		\hfill
		\begin{minipage}[b]{.43\textwidth}
			\centering
			\includegraphics[width=1\linewidth]{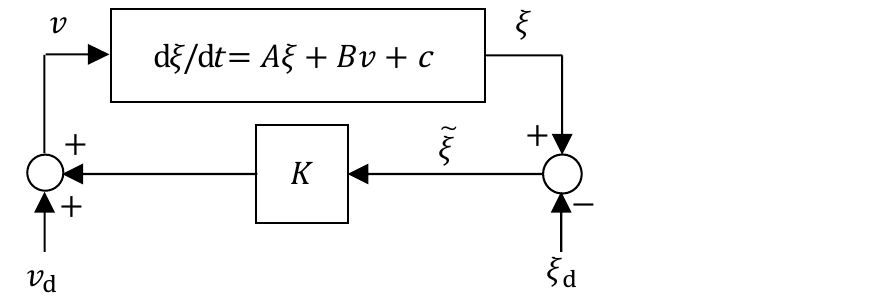}
			\caption{Controller design}
			\label{fig: Controller design}
		\end{minipage}
	\end{figure}

	\subsection{Constraint-Aware Control}
	\label{subsec:QPmethod}
	
	Thus far, constraints \eqref{cons_v} and \eqref{cons_z} have not been considered. 
	We employ the control barrier function (CBF) as one of the methods to satisfy these constraints.
	
	First, we preliminarily introduce the standard CBF approach considering an input-affine single-output nonlinear system:
	\begin{align}
		\dot{\xi} &= f( \xi) + g( \xi) v, \label{eq: CBF target system}\\
		y^\prime &= h( \xi)
	\end{align}
	with $\xi \in X ( \subset \mathbb{R}^n), v \in V ( \subset \mathbb{R}^m) , y^\prime \in \mathbb{R}, h : X \rightarrow \mathbb{R}$. 
	Note that \eqref{eq: CBF target system} corresponds to \eqref{eq:x_affine} when applied to our EL models, and $V$ corresponds to the input constraint \eqref{cons_v}.
	A continuous and differentiable function $h : X \rightarrow \mathbb{R}$ is called a CBF if a class $\mathcal{K}_\infty$ function\footnote{Function $\alpha $ is said to be of class $\mathcal{K}_\infty$ if $\alpha$ is a strictly monotonically increasing function, and $\alpha(0) = 0, \alpha(r)\rightarrow \infty (r \rightarrow \infty)$.} $\alpha : \R_{\geq 0} \rightarrow \R_{\geq 0}$ exists such that $\inf_{v \in V}\left\lbrace \mathcal{L}_f h( \xi) +  \mathcal{L}_g h( \xi) v  \right\rbrace \leq \alpha (-h(\xi)) $ holds for any $\xi\in X$, where $\mathcal{L}_f h( \xi)  := \frac{\partial h}{\partial \xi} f(\xi)$.
	We can verify that if $\xi(0)\in \{\xi \in X : h(\xi) \leq 0 \} =:\mathcal{C}$ and $v(t)$ is chosen such that $\mathcal{L}_f h( \xi) +  \mathcal{L}_g h( \xi) v \leq \alpha ( -h(\xi) )$ holds for any $t$, then $\xi(t)$ does not exit $\mathcal{C}$. 
	To employ this property for the required constraint satisfaction, we can use the solution $v^*$ for the quadratic optimization 
	\begin{align}
		\begin{cases}
		v^{\ast} =& \argmin_{v \in V}  \|v - k(\xi)\|^2 \label{OCP} \\
		\mathrm{s.t.} &  \mathcal{L}_f h(\xi) +  \mathcal{L}_g h(\xi) v 
		\leq \alpha ( -h(\xi) )
		\end{cases}
	\end{align}
	as a control input, where $k( \xi) $ is a desirable control law, such as LQR \eqref{LQR}. 
	Thus, the output constraint $y^\prime = h(\xi) \leq 0$ is satisfied. 
	In the above description, the constrained output $y^\prime$ is assumed to be scalar for simplicity; however, it can be easily extended to a vector by replicating the function $h$.
	
	Next, we consider setting $h(\xi,v)=\Xi(\xi,v)-\bar{y}=:h_y(\xi,v)$ to handle output constraint \eqref{cons_z}. 
	However, the aforementioned framework is not directly applicable to this function because of the input dependency of $\Xi$, and the above discussion is only valid if $h$ depends only on $\xi$.
	To circumvent this issue, we restrict input signals to continuously differentiable ones, such as 
	\begin{align}\label{eq:lambda_to_u}
		\dot v=\lambda,
	\end{align}
	which is not restrictive for practical applications. The augmented system can then be written as 
	\begin{align}
		\diff{}{t}
		\begin{pmatrix}
		\xi \\
		v \\
		\end{pmatrix}
		&=
		\begin{pmatrix}
		f(\xi) +g(\xi) v \\
		\lambda \\
		\end{pmatrix},\\
		y^\prime &= h(\xi,v). \label{u_z}
	\end{align}
	Considering an augmented vector $[\xi^\T,v^\T]^\T$ as a state and $\lambda$ as a new input, this system falls into the standard CBF framework because the input dependence of $h$ can be eliminated.
	The CBF obtained by this augmentation is called the integral CBF (I-CBF) \citep{Ames2021}. From a different perspective, the adoption of I-CBF allows the model to have an input dependence on the function $\Xi$.
	This motivated us to solve a problem
	\begin{align}
		\begin{cases}
		\lambda^{\ast} =& \argmin_{\lambda}  \diff{}{t}\|v - k(x)\|^2 + \beta \|\lambda\|^2  \label{OCP on lambda} \\
		\mathrm{s.t.} & \frac{\partial h}{\partial \xi}( f(\xi) + g(\xi) v)  +
		\frac{\partial h}{\partial v}\lambda 
		\leq \alpha( -h(\xi, v) )
		\end{cases},
	\end{align}
	with a positive constant $\beta$, to determine $\lambda$ as $\lambda = \lambda^*$. The first term aims at the quick convergence of $v$ to $k(\xi)$, and the second term is added to regularize the time-derivative of the input $v$. Although the input constraint $v\in V$ originally considered in \eqref{OCP} disappears in the above problem, it can be treated by combining the vector function
	\begin{align}
		h_v(\xi,v):= \mat{v - \Psi(\bar{u},\Phi^{-1}(\xi) )\\
		\Psi(\underline{u},\Phi^{-1}(\xi)) - v}
	\end{align} 
	to $h(\xi,v)$ for handling the input constraint \eqref{cons_v}. 
	In summary, this approach can handle the constraints \eqref{cons_v} and \eqref{cons_z} by setting $h(\xi,v)=[h_y(\xi,v)^\T,h_v(\xi,v)^\T]^\T$.
	
	The function to be minimized in \eqref{OCP on lambda} is quadratic with respect to $\lambda$, and it is given by
	\begin{align}\label{eq: minimization F}
		F(\lambda) & := \diff{}{t}\|v - k(\xi)\|^2 + \beta \|\lambda\|^2 \\
		& = \beta \|\lambda\|^2 + 2(v-k(\xi))^\T \lambda - 2(v-k(\xi))^\T \pdiff{k(\xi)}{\xi} (f(\xi)+g(\xi) v)\nonumber.
	\end{align}
	In addition, any constraint expressed in the form $h(\xi,v)\leq 0$ is reduced to the affine constraint for $\lambda$ in the I-CBF approach.
	Thus, the above optimization problem is a QP and can be solved quickly using many ready-made solvers.
	Further, the EL model allows functions $h(\xi,v)$ corresponding to \eqref{cons_v} and \eqref{cons_z} to be convex to $v$ for any fixed $\xi$. 
	This guarantees that $\xi,v$ at the equilibrium point of this system is the solution to the problem $\minimize_{v} \|v - k(\xi)\|^2 \  \mathrm{s.t.} \ h(\xi, v) \leq 0 $ (See Appendix~3).
	
	Note that, in practice, the above controller is implemented in a discrete-time manner like $v_{k+1} = v_k + \lambda_k \Delta t$, where $\Delta t$ is a sampling period.

	\section{Experimental Evaluation}\label{section:example}
	\begin{figure}[b]
		\centering
		\includegraphics[clip, scale=0.7]{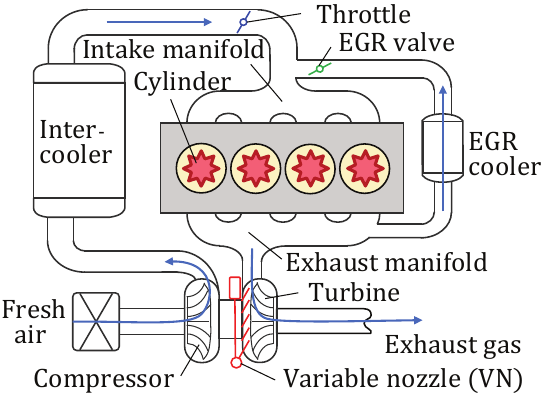}\\
		\caption{Engine air path and combustion system \label{airpath}}
	\end{figure}
	
	The modeling accuracy of the proposed model and the control results are presented using an engine air path and combustion system, as shown in Fig.~\ref{airpath}. The system included an internal combustion controller.
	In this study, we use a high-fidelity simulator, built by GT-Power and Simulink, for data generation.
	The physical quantities of each variable are as follows: The state $x:=\left[ x_1, x_2, x_3\right] $ is the net torque $\left[ \mathrm{Nm}\right] $, the NOx concentration $\left[ \log_{10}( \mathrm{ppm}) \right] $, and the generated soot quantity $\left[ \mathrm{filtered \ smoke \ number}\right] $. 
	The outputs $y:=\left[ y_1, y_2\right] $ are the combustion noise $\left[ \mathrm{dB}\right] $ and maximum cylinder pressure $\left[ \mathrm{MPa}\right] $, respectively. 
	The inputs $u:=\left[ u_1, u_2, u_3\right] $ are the variable nozzle (VN) closing degree $\left[ \mathrm{\%}\right] $, throttle closing degree $\left[ \mathrm{\%}\right] $, and exhaust gas recirculation (EGR) valve opening degree $\left[ \mathrm{\%}\right] $. The exogenous inputs $d:=\left[ d_1, d_2\right] $ are the engine speed $\left[ \mathrm{rpm}\right] $ and fuel injection rate $\left[ \mathrm{mm^3}\right]$, respectively.
	
	\begin{table}
		{\caption{Configuration for model learning\label{tab: setting}} \centering \vspace{10pt}
		\begin{tabular}{c|c} \hline
			Item  & Setting \\ \hline\hline
			$N_\Phi$ (Number of layers for $\Phi$) & 6 \\ 
			Number of neurons for $\Phi$ & 3 (for all layers) \\ \hline
			$N_\Psi$ (Number of layers for $\Psi$) & 6 \\
			Number of neurons for $\Psi$ & 3 (for all layers) \\ \hline
			$N_\Xi$ (Number of layers for $\Xi$) & 2 \\ 
			Number of neurons for $\Xi^{*1}$ & 30 (for all layers) \\ 
			Activation functions for $\Xi$  & $\varphi_{\Xi}^{(1)}(x) = {\rm softplus}(x)$, $\varphi_{\Xi}^{(2)}(x) = x$, $\varphi_{\Xi d}^{(1)}(x) = {\rm softplus}(x)$\\ \hline
			$N_d$ (Number of layers for FNN${ }^{*2}$) & 2 \\ 
			Number of neurons for FNN${ }^{*2}$ & 30 (for all layers) \\ 
			Activation functions for FNN${ }^{*2}$ & $\varphi^{(1)}(x) = \max(0.1\cdot x,x)$, $\varphi^{(2)}(x) = x$ \\ \hline
			Weight matrix $Q_e$ & {\rm diag}([1,1,1,1,1]) \\
			$N_s$ (Number of samples) & 831,488\\ \hline
		\end{tabular}\\
		\footnotesize\raggedright
		*1: The setting is applied for $\varphi_{\Xi}^{(i)},\varphi_{\Xi d}^{(i)} \ (i=1,\ldots,N_\Xi)$. 
		*2: The setting is applied for $\varphi_{\Phi d}^{(i)},\varphi_{\Psi d}^{(i)},\varphi_{\Psi x}^{(i)},\varphi_{s}^{(i)} \ (i=1,\ldots,N_d)$ in Fig.~\ref{fig: Network}. Each layers implicitly includes weights and biases of appropriate size to be learnt.}
	\end{table}
	
	The learning results are compared with those of the S--HW model \citep{Moriyasu2022}, EL model A (output $y$ is independent of the internal input $v$), and EL model B (output $y$ depends on the internal input $v$).
	The model learning settings were determined by trial and error and are listed in Table~\ref{tab: setting}. 
	The response results $\left\lbrace x, y\right\rbrace $ of each model to the input data $\left\lbrace u, d\right\rbrace $ are shown in Fig.~\ref{fig: Results of model accuracy}, and the comparison of the root mean squared error (RMSE) is listed in Table~\ref{tb:R2}.
	Fig.~\ref{fig: Results of model accuracy} shows that the response of the EL model (A, B) is qualitatively better fitted to validation data than that of the S--HW model. 
	Table~\ref{tb:R2} shows that there is a large difference in the coefficient of determination of the state $x$ between the S--HW and EL models (A, B).
	The results confirm that the modeling accuracy can be improved by adding the dependency of the internal input $v$ on the output $y$.
	The target system has complex multi-physics, and therefore, whether it is EL or not has not been verified by physical modeling.
	However, since some outputs are significantly more accurate with EL than with S--HW, it can be inferred that the target system contains characteristics that cannot be inherently represented with S--HW.
	
	\begin{figure}[p!]
		\def\@captype{table}
		\begin{minipage}[b]{.5\textwidth}
			\centering
			\includegraphics[width=0.8\linewidth]{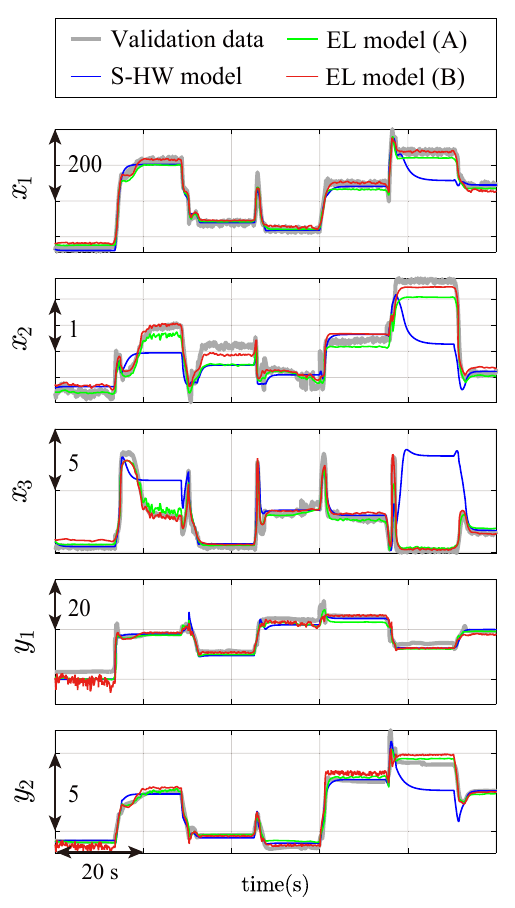}
			\figcaption{Results of model prediction\label{fig: Results of model accuracy}}
			\vspace{10pt}
			\includegraphics[width=0.95\linewidth]{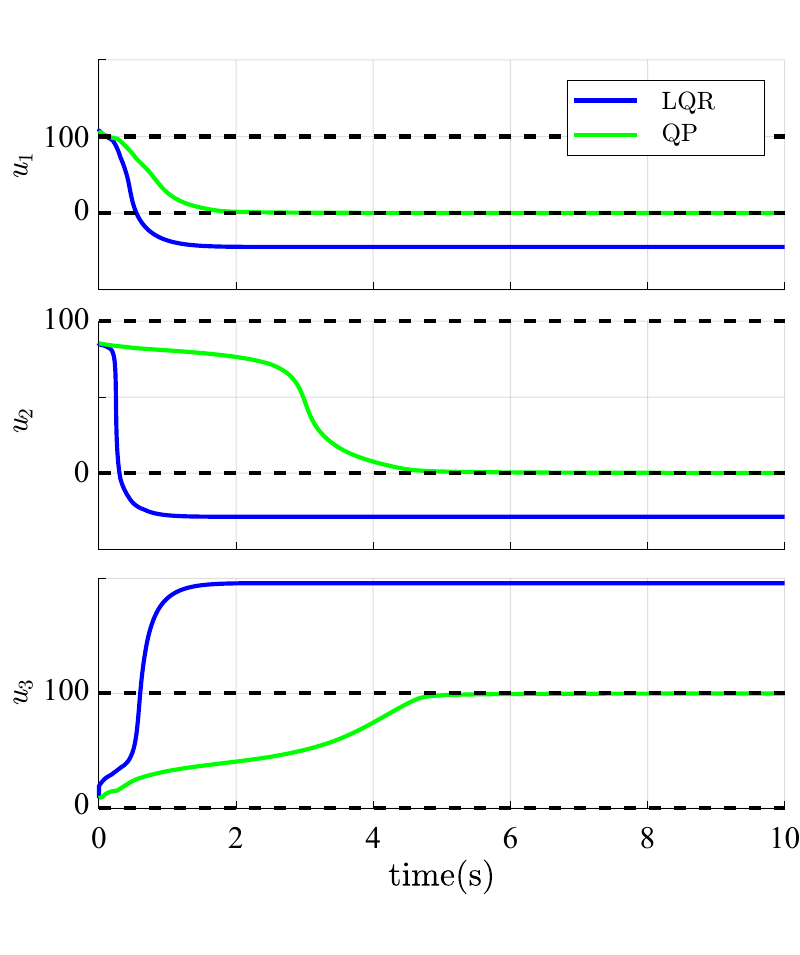}\\
			\caption{Control result of input $u$ \label{fig: Control result of input}}
		\end{minipage}
		\hfill
		\begin{minipage}[b]{.48\textwidth}
			\centering
			\begin{center}
				\tblcaption{Root mean squared error}\label{tb:R2}\vspace{10pt}
				\begin{tabular}{c|c|c|c} \hline
					      & S--HW & EL(A) & EL(B) \\ \hline\hline
					$x_1$ & 17.72 & 16.06 & 15.49 \\ \hline
					$x_2$ & 0.245 & 0.189 & 0.140 \\ \hline
					$x_3$ & 1.218 & 0.546 & 0.573 \\ \hline
					$y_1$ & 2.668 & 2.600 & 2.140 \\ \hline
					$y_2$ & 0.505 & 0.404 & 0.377 \\ \hline
				\end{tabular}
			\end{center}
			\vspace{20pt}
			\vfill
			\includegraphics[width=.95\linewidth]{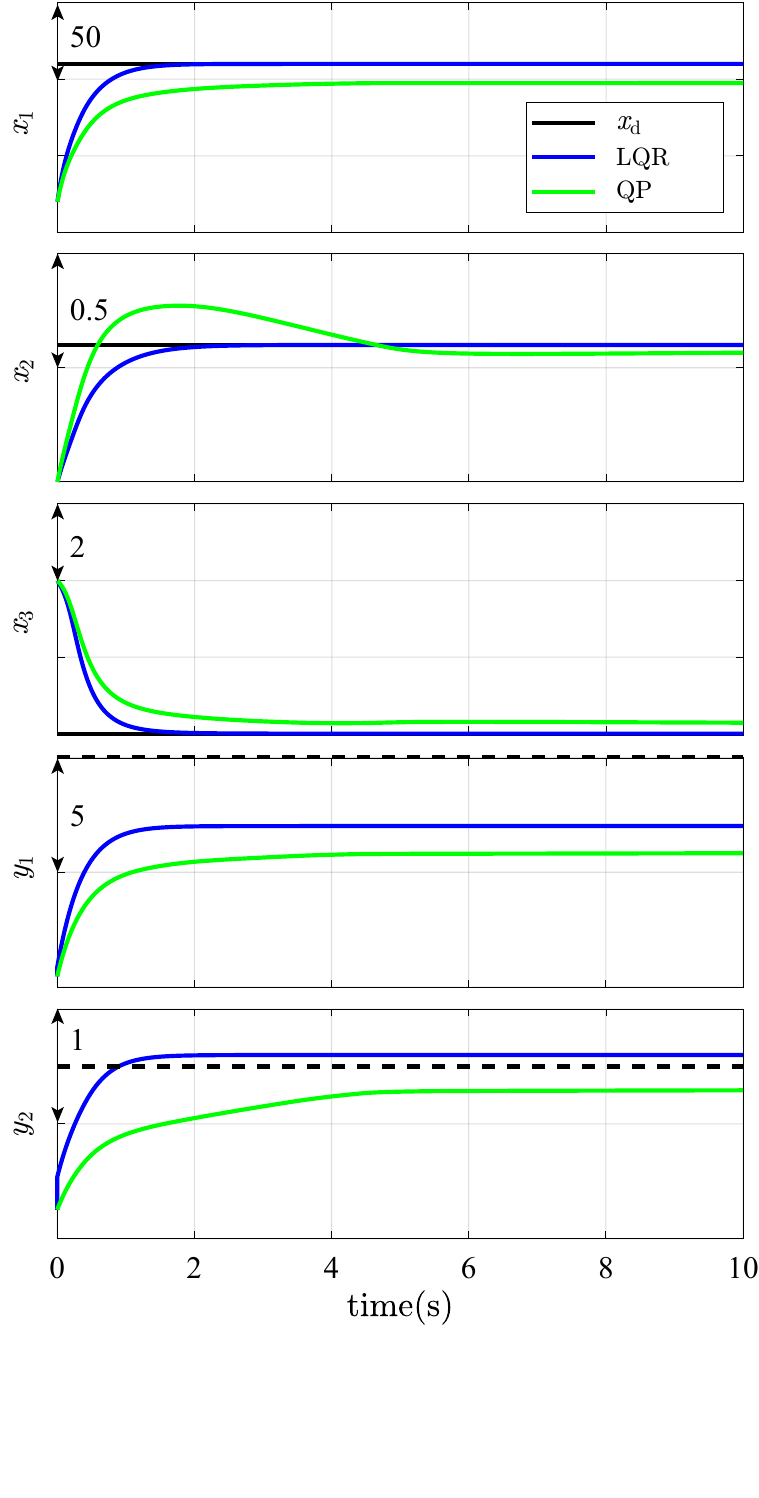}
			\caption{Control result of state $x$ and output $y$}
			\label{fig: Control result of output}
		\end{minipage}
	\end{figure}

	Next, we show the control simulation results.
	The learned EL model (B) is used as the control model and also as the system to be controlled; i.e., assuming no modeling error; since the main focus is on the nominal behavior of the proposed methods.
	We designed two types of controllers using the methods described in \ref{subsec:LQR} and \ref{subsec:QPmethod}.
	That is, the former one is LQR \eqref{LQR} and the later solves QP \eqref{OCP on lambda} to calculate the control input online using LQR \eqref{LQR} as the reference function $k(\xi)$. 
	Weight matrices are set to be $Q={\rm diag}([0.1,0.1,0.1]), R={\rm diag}([10,10,10])$ and regularization parameter is $\beta = 10$.
	These were determined by trial and error.
	The sampling time is 0.001~s, the input and output constraints represented by \eqref{cons_v} and \eqref{cons_z} are $0 \leq u_i \leq 100\, ( i =  1, \ldots, 3 ) $ and $y_1 \leq 80, y_2 \leq 6.5 $. 
	In addition, the class $\mathcal{K}_\infty$ functions $\alpha_i, \alpha_j\, (i = 1, \ldots 3, j = 1, 2)$ in \eqref{OCP on lambda} are quadratic. 
	The numerical results are shown in Figs.~\ref{fig: Control result of input} and \ref{fig: Control result of output}. 
	In these figures, each result and color scheme are shown below.
	\begin{itemize} 
		\item Reference $x_\mathrm{d}$ for the state $x$ (black solid), 
		\item Bounds for the input $u$ and output $y$ (black dashed),
		\item Result of unconstrained control by LQR \eqref{LQR} (blue), 
		\item Result of constrained control by QP \eqref{OCP on lambda}  (green). 
	\end{itemize}
	We can see that QP realized the regulation control of state $x$ to the reference $x_\mathrm{d}$ while satisfying constraints \eqref{cons_v} and \eqref{cons_z}.
	For input $u$, the control law based on \eqref{LQR} (blue line) fails to satisfy the constraint \eqref{cons_v} (black dashed line), whereas the control law based on QP (green line) is within the constraints.
	For state $x$, the results of the QP deviate slightly from the reference $x_\mathrm{d}$ because there is no input $u$ such that $v = \Psi(u, x,\bar{d})$ satisfies constraints \eqref{cons_v} and \eqref{cons_z} under $x = x_\mathrm{d}$. 
	In other words, from \eqref{eq: target value on LQR}, it is possible to obtain $\xi_\mathrm{d}$ and $v_\mathrm{d}$ corresponding to $x_\mathrm{d}$; however, $u, y$ obtained by $\Psi$ in \eqref{eq: v to u} and $\Xi$ in \eqref{eq: x,u to z} do not satisfy the constraints \eqref{cons_v}, \eqref{cons_z}.
	Fig.~\ref{fig: Control result of output} suggests that our proposed method automatically find a suitable steady input $u$ that achieves good regulation performance for the target $x_\mathrm{d}$ while satisfying constraints \eqref{cons_v} and \eqref{cons_z}, which is the purpose of this study.
	The computation time of the proposed contoller based on QP was 0.396~ms for each control period on Intel Core i7-1265U processor, showing sufficiently fast computation. We used QP-KWIK solver \citep{schmid1994quadratic} to solve QP \eqref{OCP on lambda}.

	In practice, we cannot assume that there is no modeling error and the target system have to be replaced into an actual plant or a high-fidelity simulator from which learning data have been obtained.
	In cases involving modeling errors, it is effective to configure a servo system that includes a disturbance observer and/or an integrator, but this is not the main focus of this paper and is omitted here.

	\section{Conclusion}\label{section: conclusion}
	We extended the S--HW model proposed by \citet{Moriyasu2022} in terms of modeling and control design by proposing the output feedback for the input transformation $\Psi$ in modeling and by introducing an integral control barrier function \citep{Ames2021} in the control. The experimental validation confirmed the improvement in modeling accuracy, which led to better control. However, the model structure proposed in this study did not cover the full range of linearizable classes. Our future work will cover a wider range of representational capabilities within the class and extensions to a wider class of dynamical systems that encompass the exactly linearizable class.
	


	\section*{Notes on contributors}
	\textbf{Ryuta Moriyasu} received his B.E. and M.E. degrees from Osaka University in 2011 and 2013, and his Ph.D. degree from Kyoto University in 2024.
	Since 2013, he has been with Toyota Central R\&D Labs., Inc., Aichi, Japan.
	His research interests include machine learning, control theory, optimization theory, and their applications in the automotive field.
	Recently, he has been working on theoretical reliability assurance in control design using machine learning models, with the aim of practical application to engine control system.
	He received Best Presentation Award of the Multi-Symposium on Control System from SICE Control Division, Excellent Presentation Award from RSJ, The Asahara Science Award from JSAE, and so on.
	He is a Member of SICE and JSAE.\\

	\noindent
	\textbf{Masayuki Kusunoki} received his B.E. and M.E. degrees from Shizuoka University in 2013 and 2015.
	Since 2015, he has been with Toyota Industries Corporation, Aichi, Japan.
	He has been engaging in control development of diesel engines for passenger cars.\\

	\noindent
	\textbf{Kenji Kashima} received his Doctoral degree in Informatics from Kyoto University in 2005. He was with Tokyo Institute of Technology, Universität Stuttgart, Osaka University, before he joined Kyoto University in 2013, where he is currently an Associate Professor. His research interests include control and learning theory for complex dynamical systems, and their applications. He received Humboldt Research Fellowship (Germany), IEEE CSS Roberto Tempo Best CDC Paper Award, Pioneer Award of SICE Control Division, and so on. He has served as an Associate Editor of IEEE Transactions of Automatic Control, IEEE CONTROL SYSTEMS LETTERS, and IEEE CSS Conference Editorial Board and a steering committee member of MTNS. He is a Senior Member of IEEE and Member of SICE, ISCIE and IEICE.\\

	
	\bibliography{MyCollection}
	
	
	\section*{Appendix~1. Exactly Linearizable Condition \citep{Su1982}}
	\label{section: Exact linearizability condition}
	
	In this section, we provide the necessary and sufficient conditions for exact linearizability introduced in Section \ref{subsec:exact_linearization}. The system of interest is represented by 
	\begin{align}
		\frac{dx}{dt} = \overline{f}(x) + \overline{g}(x)u. \label{eq: SI exact linearizable equation}
	\end{align}
	However, we consider a single-input system ($u \in \R$) for simplicity.
	Exact linearization implies that we can transform \eqref{eq: SI exact linearizable equation} into the following linear system by applying a coordinate transformation \eqref{eq: coordinate change} and feedback~\eqref{eq: exact linearization feedback}
	\begin{align} \label{eq: SI linearized system}
		\frac{d\xi}{dt} = \begin{pmatrix}
		0 & 1 & 0 & \cdots &  0\\
		0 & 0 & 1 & \cdots &  0\\
		\vdots & \vdots &  &\ddots & 0 \\
		0 & 0 & 0 & \cdots &  1\\
		0 & 0 & 0 & \cdots &  0
		\end{pmatrix}
		\xi + 
		\begin{pmatrix}
		0\\
		0\\
		\vdots \\
		0\\
		1
		\end{pmatrix}
		v.
	\end{align}
	In addition, we define two notations.
	\begin{dfn}
		From \eqref{eq: SI exact linearizable equation}, for functions $\overline{f}, \overline{g}$, we define the operator $\mathbf{ad}_{\overline{f}} \overline{g}$ as
		We also refer to $\left[\overline{f}, \overline{g}\right]$ as the lie brackets.
		\begin{align}
		\mathbf{ad}_{\overline{f}} \overline{g}&:= 
		\left[\overline{f}, \overline{g}\right], \label{eq: operator}\\
		\left[\overline{f}, \overline{g}\right]& := 
		\frac{\partial \overline{g}}{\partial x} \overline{f}(x) - \frac{\partial \overline{f}}{\partial x} \overline{g}(x).
		\end{align}
		Subsequently, the operator \eqref{eq: operator} acts as 
		\begin{align}
		\mathbf{ad}^{k + 1}_{\overline{f}} \overline{g} &= \left[\overline{f},\mathbf{ad}^{k}_{\overline{f}} \overline{g} \right],\\
		\mathbf{ad}^{0}_{\overline{f}} \overline{g} &=\overline{g}(x),\\
		\mathcal{L}_{\mathbf{ad}_{\overline{f}} \overline{g}} \phi(x) &= \mathcal{L}_{\overline{f}}\mathcal{L}_{\overline{g}}\phi(x) - \mathcal{L}_{\overline{g}}\mathcal{L}_{\overline{f}}\phi(x),
		\end{align}
		where $\phi(x)$ denotes the arbitrary function $x$.
	\end{dfn}
	
	The necessary and sufficient conditions for exact linearizability can be expressed as follows:
	\begin{thm}
		A necessary and sufficient condition for the existence of coordinate transformation \eqref{eq: coordinate change} and feedback \eqref{eq: exact linearization feedback} that transform the equation of state \eqref{eq: SI exact linearizable equation} into a linear system \eqref{eq: SI linearized system} is the simultaneous satisfaction of the following two conditions:
		\begin{enumerate}
		\item $\left\lbrace \mathbf{ad}^{0}_{\overline{f}} \overline{g},\mathbf{ad}^{1}_{\overline{f}} \overline{g}, \cdots , \mathbf{ad}^{n-1}_{\overline{f}} \overline{g} \right\rbrace(x)$ is linearly independent for any $x$.\label{EL condition1}\\
		\item $\left\lbrace \mathbf{ad}^{0}_{\overline{f}} \overline{g},\mathbf{ad}^{1}_{\overline{f}} \overline{g}, \cdots , \mathbf{ad}^{n-2}_{\overline{f}} \overline{g} \right\rbrace(x)$ is involutive.\label{EL condition2}
		\end{enumerate}
	\end{thm}\noindent
	In this theorem, \ref{EL condition1} is a condition for system controllability, and \ref{EL condition2} provides the necessary and sufficient conditions for the unique existence of the coordinate transformation \eqref{eq: coordinate change} and feedback \eqref{eq: exact linearization feedback} when \ref{EL condition1} is satisfied by the Frobenius' theorem \citep{Su1982}.
		
	\section*{Appendix~2. Nonlinear Optimal Control for Exactly Linearizable Input-Affine Dynamics}
	\label{App:NOC}
	
	We proposed a control design method that fully utilizes simple linear optimal control theory. In this section, we show how to utilize nonlinear optimal control theory based on the modeling framework in the present study, assuming that the system to be controlled is an input-affine nonlinear system. 
	
	Let us consider the following input-affine nonlinear system, as in \eqref{eq: coordinate change} and \eqref{eq: exact linearization feedback}.  
	We can construct $\Psi_1,\ \Psi_2,\ A,\ B,\ c$ as
	\begin{align}
		u(t) &= \Psi_1(x(t))+\Psi_2(x(t))v(t),  \label{app:u}\\
		\dot{\xi}(t) &= A \xi(t) + B v(t) + c, \label{app:dx}\\
		x(t) &= \Phi^{-1}(\xi(t)), \label{app:y}
	\end{align}
	using the data for $\{u,x,\dot x\}$, 
	where $\Phi$ is bijective and $\Psi_2(x)$ is nonsingular for any $x$. Without the loss of generality, we assume $\Phi(0) = 0$ and $c = 0$. 
	The obtained dynamics are input-affine in the form of \eqref{eq: exact linearizable equation}. 
	Next, we define the control Lyapunov function (CLF), which is essential for explaining the Sontag-type stabilizing control law.
	\begin{dfn}[\cite{Ames2019}]
		A positive definite function $V: Y\rightarrow \R$ is called a CLF if there exists a class $\mathcal{K}$ function $\gamma$ such that 
		\begin{align}
		\inf_{v \in D} \left\lbrace \mathcal{L}_{\bar f} V( x) +  \mathcal{L}_{\bar g} V(x) u \right\rbrace <  - \gamma( V(x) ). 
		\end{align}
	\end{dfn}
	For our model, we explicitly provide the CLF
	\begin{align}
		{V}(x) = \Phi( x)^\T P\Phi( x) \label{df_CLF}
	\end{align}
	with the solution $P$ to \eqref{eq: riccati equation}. 
	Once a CLF is found, we can utilize several tools from the nonlinear control theory. 
	\begin{thm}[\cite{Sontag1998}]\label{thm:sontag} 
		Let a positive definite function $V: Y \rightarrow \R$ be a CLF for system \eqref{eq: exact linearizable equation}; then, the following Sontag-type control law $v = \alpha_d( x) $ asymptotically stabilizes the origin of the system \eqref{eq: exact linearizable equation}.
		\begin{align}
		v& = \alpha_d( x)\label{cl_sontag} \\
		\alpha_d( x)  &:= 
		\begin{cases}
		-\frac{1}{r_d( x) }( \mathcal{L}_g V) ^\T, \ \ &( \mathcal{L}_g V \neq 0) \\
		0, &  ( \mathcal{L}_g V = 0)
		\end{cases}
		\end{align}
		where the bounded function $r_d( x) >0$ is given by
		\begin{align}
		r_d( x) := 
		\frac{\mathcal{L}_g V ( \mathcal{L}_g V ) ^\T}
		{\mathcal{L}_f V + \sqrt{\mathcal{L}_f V ^2 + (\mathcal{L}_g V ( \mathcal{L}_g V ) ^\T) ^2 }}.
		\end{align}
	\end{thm}
	\noindent
	This control law not only stabilizes, but also satisfies the inverse optimality for nonlinear dynamics \citep{Khalil2002}. This property makes the control law robust, as in the phase margin of the linear control theory. For a more extended framework, please refer to differential flatness \citep{Fliess1995}.
	
	\section*{Appendix~3. Equilibrium Condition of Proposed Control System}
	
	The system proposed in Section~3.2 can be summarized as
	\begin{align}
		&\diff{}{t}
		\begin{pmatrix}
			\xi \\
		v \\
		\end{pmatrix}
		=
		\begin{pmatrix}
		f(\xi) + g(\xi) v \\
		\lambda \\
		\end{pmatrix}, \label{AugSys}\\
		&\begin{cases}
		\lambda^{\ast} =& \argmin_{\lambda}  \diff{}{t}\|v - k(\xi)\|^2 + \beta\|\lambda\|^2  \label{OCPrecall} \\
		\mathrm{s.t.} & \frac{\partial h}{\partial \xi}( f( \xi) +g( \xi) v)  +
		\frac{\partial h}{\partial v}\lambda 
		\leq \alpha( -h(\xi, v) )
		\end{cases},
	\end{align}
	where $\beta>0$ and $h(\xi,v)$ are convex with respect to $v$, for any fixed $\xi$.
	The convexity originated from the structure of the EL model.
	We can prove the optimality of the equilibrium as 
	
	\begin{thm}
		If an equilibrium point $(\xi,v)$ exists in the system shown in \eqref{AugSys} and \eqref{OCPrecall}, then the points satisfy $v = \argmin_v \| v - k(\xi) \|^2 \ \mathrm{s.t.} \ h(\xi,v) \leq 0$.
	\end{thm}
	\begin{proof}
		As the problem in \eqref{OCPrecall} is strongly convex with $\beta > 0$, the necessary and sufficient condition for first-order optimality can be written as
		\begin{align}
			& \nabla F(\lambda) + \PAR{\nabla G(\lambda)}^\T \mu = 0,\\
			& \mu \geq 0, \ G(\lambda) \leq 0, \ \mu \circ G(\lambda) = 0,
		\end{align}
		where $F(\lambda):= \diff{}{t}\|v - k(\xi)\|^2 + \beta \|\lambda\|^2, \ G(\lambda):= \frac{\partial h}{\partial \xi}( f( \xi) +g( \xi) v)  + \frac{\partial h}{\partial v}\lambda - \alpha( -h(\xi, v) )$, and $\mu$ represents a Lagrange multiplier. 	
		Considering the equilibrium condition of \eqref{AugSys}, i.e., $f(\xi)+g(\xi)v=0, \ \lambda = 0$, the above optimality condition reduces to
		\begin{align}
			& 2(v-k(\xi)) + \PAR{\pdiff{h}{v}}^\T \mu = 0,\\
			& \mu \geq 0, \ h(\xi,v) \leq 0, \ \mu \circ h(\xi,v) = 0.
		\end{align}
		$h(\xi,v)$ is convex to $v$, and therefore, the above condition coincides with the necessary and sufficient condition of the first-order optimality of the optimization problem 
		\begin{align}
			\minimize_{v}  \|v - k(\xi)\|^2 \ \mathrm{s.t.} \ h(\xi,v) \leq 0. \label{OCPsteady}
		\end{align}	
	\end{proof}

	\section*{Word Count}
	This manuscript includes 4424 words.

	\end{document}